\documentclass[journal]{IEEEtran}
\usepackage{amsmath}
\usepackage{amssymb}
\usepackage{amsfonts}
\usepackage{graphicx}
\usepackage{epsfig}
\usepackage{subfigure}
\usepackage{psfrag}
\usepackage{color}

\begin{document}
\title{Throughput Optimal Policies for Energy Harvesting Wireless Transmitters with Non-Ideal Circuit Power}

\author{Jie Xu and Rui Zhang
\thanks{J. Xu is with the Department of
Electrical and Computer Engineering, National University of
Singapore (e-mail:elexjie@nus.edu.sg).}
\thanks{R. Zhang is with the Department of Electrical and Computer Engineering, National
University of Singapore (e-mail: elezhang@nus.edu.sg). He is also
with the Institute for Infocomm Research, A*STAR, Singapore.}}

\maketitle

\begin{abstract}\label{sec:abstract}
Characterizing the fundamental tradeoffs for maximizing \emph{energy
efficiency} (EE) versus \emph{spectrum efficiency} (SE) is a key
problem in wireless communication. In this paper, we address this
problem for a point-to-point additive white Gaussian noise (AWGN)
channel with the transmitter powered solely via \emph{energy
harvesting} from the environment. In addition, we assume a practical
on-off transmitter model with \emph{non-ideal} circuit power, i.e.,
when the transmitter is on, its consumed power is the sum of the
transmit power and a constant circuit power. Under this setup, we
study the optimal transmit power allocation to maximize the average
throughput over a finite horizon, subject to the time-varying energy
constraint and the non-ideal circuit power consumption. First, we
consider the \emph{off-line} optimization under the assumption that
the energy arrival time and amount are \emph{a priori} known at the
transmitter. Although this problem is non-convex due to the
non-ideal circuit power, we show an efficient optimal solution that
in general corresponds to a \emph{two-phase} transmission: the first
phase with an \emph{EE-maximizing} on-off power allocation, and the
second phase with a \emph{SE-maximizing} power allocation that is
non-decreasing over time, thus revealing an interesting result that
both the EE and SE optimizations are unified in an energy harvesting
communication system. We then extend the optimal off-line
algorithm to the case with multiple parallel AWGN channels, based on
the principle of \emph{nested optimization}. Finally, inspired
by the off-line optimal solution, we propose a new \emph{online}
algorithm under the practical setup with only the past and present
energy state information (ESI) known at the transmitter.
\end{abstract}

\begin{keywords}
Energy harvesting, power control, energy efficiency, spectrum
efficiency, circuit power.
\end{keywords}

\IEEEpeerreviewmaketitle
\setlength{\baselineskip}{1.0\baselineskip}
\newtheorem{definition}{\underline{Definition}}[section]
\newtheorem{fact}{Fact}
\newtheorem{assumption}{Assumption}
\newtheorem{theorem}{\underline{Theorem}}[section]
\newtheorem{lemma}{\underline{Lemma}}[section]
\newtheorem{corollary}{\underline{Corollary}}[section]
\newtheorem{proposition}{\underline{Proposition}}[section]
\newtheorem{example}{\underline{Example}}[section]
\newtheorem{remark}{\underline{Remark}}[section]
\newtheorem{algorithm}{\underline{Algorithm}}[section]
\newcommand{\mv}[1]{\mbox{\boldmath{$ #1 $}}}

\section{Introduction}\label{sec:introduction}

\PARstart{G}{reen} or energy efficient wireless communication has recently drawn
significant attention due to the growing concerns about the
operator's cost as well as the global environmental cost of using
fossil fuel based energy to power cellular infrastructures
\cite{Chen:NetworkES,Han:GreenRadio,Niu:CellZooming}. To achieve the
optimal energy usage efficiency for cellular networks, various
innovative ``green'' techniques across different layers of
communication protocol stacks have been proposed
\cite{Bhargava:GreenCellular,YChenComMag}. Among others, how to
maximize the bits-per-Joule {\it energy efficiency} (EE) for the
point-to-point wireless link has received a great deal of interest
recently
\cite{Energy_EfficientPacketTransmission,Miao,fractionalprogramming}.

\begin{figure}
\centering
 \epsfxsize=1\linewidth
    \includegraphics[width=7cm]{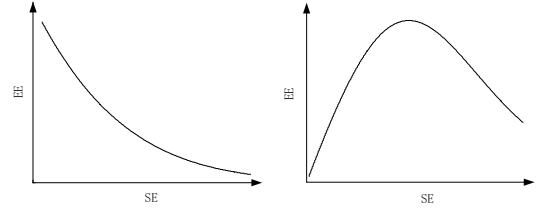}
\caption{Tradeoff between EE and SE for the ideal circuit power case
of $\alpha = 0$ (the left sub-figure) and the non-ideal circuit
power case of $\alpha > 0$ (the right sub-figure).} \label{fig:EESE}
\end{figure}

Besides maximizing EE, another key design objective in wireless
communication is to maximize the {\it spectrum efficiency} (SE) or
the number of transmitted bits-per-second-per-Hz (bps/Hz), due to
the explosive growth of wireless devices and applications that
require high data rates. In order to design wireless communication
systems both energy and spectrum efficiently, the fundamental EE-SE
relationship needs to be examined carefully. For the simple additive
white Gaussian noise (AWGN) channel with bandwidth $W$ and noise
power spectral density $N_0$, by applying  the Shannon's capacity
formula, the SE and EE are expressed as $\xi_{\rm SE} =
\log_2(1+\frac{P}{WN_0})$ and $\xi_{\rm EE} =
W\log_2(1+\frac{P}{WN_0})/P$, respectively, with $P$ denoting the
transmit power. It thus follows that the optimal EE-SE tradeoff is
characterized by $\xi_{\rm EE}=\frac{\xi_{\rm SE}}{(2^{\xi_{\rm
SE}}-1)N_0}$, where $\xi_{\rm EE}$ is a monotonically decreasing
function of $\xi_{\rm SE}$, as shown in the left sub-figure of Fig.
\ref{fig:EESE}. In this case, any SE increment will inevitably
result in a decrement in EE. However, in practical wireless
transmitters, besides the direct transmit power $P$, there
also exists {\it non-ideal} circuit power consumed when $P>0$, which
accounts for the power consumptions at e.g. the AC/DC converter and
the analog radio frequency (RF) amplifier, and amounts to a
significant part of the total consumed power at the transmitter.
Moreover, when there is no data transmission, i.e., $P=0$, the
transmitter can turn into a {\it micro-sleep} mode \cite{Blume:ES},
by switching off the power amplifier to reduce the circuit power
consumption. For the ease of description, in this paper the transmitter status
with $P>0$ and $P=0$ are referred to as the \emph{on} and \emph{off}
modes, respectively. Denote the non-ideal circuit power during an
``on'' mode as $\alpha\geq 0$ in Watt, the efficiency of the RF chain as
$0<\eta\leq 1$, and the power consumed during an ``off'' mode as
$\beta\geq 0$ in Watt. A practical power consumption model for the
wireless transmitter is given by \cite{Kim}
\begin{eqnarray} \label{eq1}
P_{\rm{total}}=\left\{\begin{array}{ll} \frac{P}{\eta} + \alpha, &
P>0 \\ \beta, & P=0, \end{array} \right.
\end{eqnarray}
where $P_{\rm{total}}$ denotes the total power consumed at the
transmitter. In practice, $\beta$ is generally much smaller as
compared to $\alpha$ and thus can be ignored for simplicity
\cite{YChenComMag,Miao,fractionalprogramming}. In this paper, we
assume $\beta = 0$. Therefore, in (\ref{eq1}) without
loss of generality we can further assume $\eta = 1$ since $\eta$ is
only a scaling constant.{\footnote{Note that the results
of this paper can be readily extended to the case with
$\eta < 1$ by appropriately scaling the obtained solutions.}}
With the above simplifications, the EE can be
re-expressed as $\xi_{\rm EE} =
W\log_2(1+\frac{P}{WN_0})/(P+\alpha)$ for $P>0$ and the resulting new EE-SE
tradeoff is shown in the right sub-figure of Fig. \ref{fig:EESE} for
a given $\alpha>0$, from which it is observed that the non-ideal circuit
power drastically  changes the behavior of the EE-SE tradeoff as
compared to the ideal case of $\alpha=0$.

Recently, a new design paradigm for achieving green wireless
communication has drawn a great deal of attention, in which
wireless terminals are powered primarily or even solely by
harvesting the energy from environmental sources such as solar and wind,
thereby reducing substantially the
energy cost in traditional wireless systems
\cite{Bhargava:GreenCellular,Kansal}. With the embedded energy
harvesting device and rechargeable battery, wireless transmitters
can replenish energy from the environment without the need of
replacing battery or drawing power from the main grid. Thus,
communication utilizing energy harvesting nodes can promisingly
achieve a jointly spectrum and energy efficiency maximization goal.
However, there are new challenges in designing energy harvesting
powered wireless communication, which are not present in traditional
systems. For example, the intermittent nature of most practical energy
harvesting sources causes random power availability at the
transmitter, due to which a new type of transmitter-side power
constraint, namely {\it energy harvesting constraint}, is
imposed, i.e., the energy accumulatively consumed up to any time
cannot exceed that accumulatively harvested. As a result, existing
EE-SE tradeoffs (cf. Fig. \ref{fig:EESE}) revealed for conventional wireless systems assuming
a given constant power supply  are not directly
applicable to an energy harvesting system, with or without the
non-ideal circuit power. It is worth noting that some prior work in the
literature has investigated the throughput-optimal power control
policies for the energy harvesting wireless transmitter assuming an
ideal circuit-power model (i.e., $\alpha=0, \eta=1$, and $\beta=0$ in (\ref{eq1})), in
which useful structural properties of the optimal solution were
obtained (see e.g. \cite{Yang,Zhang,Yener} and
references therein). However, there is very limited work on studying
the effects of the non-ideal circuit power with $\alpha>0$ on the
throughput-optimal power allocation for energy harvesting communication
systems. To our best knowledge, only \cite{Bai2011} has proposed a
calculus-based approach to address this
problem; however, it does not reveal the structure of the optimal
solution. Motivated by the known result
that the non-ideal circuit power modifies the EE-SE tradeoff
considerably in the conventional case with constant power supply as
shown in  Fig. \ref{fig:EESE}, we expect that it should also play an
important role in the EE-SE tradeoff characterization under the new
setup with random power supply due to energy harvesting, which motivates our work.

In this paper, we study the throughput maximization problem for a
point-to-point AWGN channel with an energy harvesting powered
transmitter over a finite horizon. For the purpose of exposition, we
assume that the receiver has a constant power supply (e.g. battery).
We also assume that at the transmitter, the renewable energy arrives
at a discrete set of time instants with variable energy amount. Under this setup, we
investigate the effects of the
non-ideal circuit power with $\alpha>0$ on the throughput-optimal power allocation
as well as the resulting new EE-SE tradeoff. The main contributions of
this paper are summarized as follows.
\begin{itemize}
\item First, we consider the \emph{off-line} optimization under the assumption that the energy
arrival time and amount for harvesting are \emph{a priori} known at
the transmitter. We show that the optimal power allocation to
maximize the average throughput under this setup is a {\it
non-convex} optimization problem, due to the non-ideal circuit
power. Nevertheless, we derive an efficient optimal
solution for this problem, which is shown to correspond to a novel
\emph{two-phase} transmission structure: the first phase with an
\emph{EE-maximizing} on-off power allocation, and the second phase
with a \emph{SE-maximizing} power allocation that is non-decreasing
over time. Thus, we reveal an interesting result that both the EE
and SE optimizations are unified in an energy harvesting powered
wireless system.

\item We then extend the optimal off-line policy for the single-channel case to
the general case with multiple parallel AWGN channels, subject to a
total energy harvesting power constraint. Using tools from {\it nested optimization}, we transform
this problem with multi-dimensional (vector) power optimization to
an equivalent one with only one-dimensional (scalar) power
optimization, which can then be efficiently solved by the algorithm
derived for the single-channel case.

\item Furthermore, inspired by the off-line optimal solution, we propose a heuristic
\emph{online} algorithm under the practical setup where only the
causal (past and present) energy state information (ESI) for
harvesting is assumed to be known at the transmitter. It is shown by
simulations that the proposed online algorithm achieves a small
performance gap from the throughput upper bound by the optimal off-line
solution, and also outperforms other heuristically designed online
algorithms.
\end{itemize}

The rest of this paper is organized as follows. Section
\ref{sec:system model} introduces the system model and presents the
problem formulation. Section \ref{sec:offline} derives the optimal
off-line power allocation policy for the single-channel case.
Section \ref{sec:multi-channel} extends the result to the
multi-channel case based on the nested optimization. Section
\ref{sec:online} presents the proposed online algorithm and Section
\ref{sec:simulation} evaluates its throughput performance by simulations.
Finally, Section \ref{sec:conclusion} concludes the paper.

\section{System Model and Problem Formulation}\label{sec:system model}

In this paper, we consider the point-to-point transmission over an
AWGN channel with constant channel and coherent detection at the
receiver. The transmitter is assumed to replenish energy from an
energy harvesting device that collects energy over time from a
renewable source (e.g. solar or wind). We consider the block-based
energy scheduling with each block spanning over $T$ seconds (secs). We
assume that the renewable energy arrives during each block at $N-1$
time instants given by $0<t_{1}<\cdots< t_{N-1}<T$, and the energy
values collected at these time instants are denoted by
$E_1,\ldots,E_{N-1}$, respectively. In general, $N\geq 1$, $t_i$,
and $E_i>0$, $i=1,\ldots,N-1$, are modeled by an appropriate random
process for the given energy source. For convenience, we assume $t_0
= 0$ and denote $E_0$ as the initial energy stored in the energy storage
device at the beginning time of each block. For the purpose of
exposition, we assume that the energy storage device has an infinite
capacity in this paper. Moreover, we refer to the time interval
between two consecutive energy arrivals as an {\it epoch}, and
denote the length of the $i$th epoch as $L_i = t_{i} - t_{i-1},
i=1,\ldots,N$; for convenience, we denote $t_N = T$.

Suppose that the transmit power over time in each block is denoted
by $P(t)\geq 0, t\in (0,T]$. Assume that the maximum transmission
rate that can be reliably decoded at the receiver at any time $t$ is
a function of $P(t)$, given by $C(t) = R(P(t))$, which satisfies the
following properties:
\begin{enumerate}
\item $R(P(t)) \ge 0$, $\forall P(t)\geq 0$, and $R(0)=0$;
\item $R(P(t))$ is a strictly concave function over $P(t)\geq 0$;
\item $R(P(t))$ is a monotonically increasing function over $P(t)\geq 0$.
\end{enumerate}
For example, if adaptive modulation and coding (AMC) is applied at the transmitter, then the
achievable rate $C(t)$ is denoted by \cite{GoldsmithBook}
\begin{equation} \label{sys1}
\begin{array}{l}
\displaystyle R(P(t)) = W\log_2\left(1+\frac{hP(t)}{\Gamma WN_0}\right)
\end{array}
\end{equation}
in bits-per-sec (bps), where $\Gamma$ accounts for the gap from the channel
capacity due to a practical coding and modulation scheme used; $h>0$ denotes
the constant channel power gain.

As discussed in Section \ref{sec:introduction}, we assume an on-off
transmitter power model given in (\ref{eq1}) with $\beta = 0$ and
$\eta =1$; thus, we rewrite (\ref{eq1}) as
\begin{equation} \label{eq102}
\begin{array}{l}{P_{{\rm{total}}}}(t) = \left\{ {\begin{array}{*{20}{c}}
\displaystyle {{P(t)} + {\alpha},}\\
{{0},}
\end{array}} \right.\begin{array}{*{20}{l}}
{P(t) > 0}\\
{P(t) = 0}.
\end{array}\end{array}
\end{equation}

Since the accumulatively consumed energy up to any time at the
transmitter cannot exceed the energy accumulatively harvested, the
energy harvesting constraint on the total consumed power is given by
\begin{equation} \label{eq103}\begin{array}{l}\displaystyle \int_0^{t_i} {{P_{{\rm{total}}}}\left( t \right){\rm{d}}t}  \le \sum\limits_{j = 0}^{i - 1} {{E_j}}, ~ i = 1,\ldots,N. \end{array}
\end{equation}
Thus, the throughput maximization problem over a finite horizon $T$
can be formulated as follows.
\begin{align} \label{eq105}
\mathop{\max }\limits _{P(t)\geq 0} ~ &
\int_0^T {R\left( P(t) \right)} {\rm{d}}t \nonumber \\
{\rm s.t.} ~ & \int_0^{t_i} {{P_{{\rm{total}}}}\left( t \right)}
{\rm{d}}t \le \sum\limits_{j = 0}^{i - 1} {{E_j}}, ~i = 1,\ldots,N.
\end{align}

The optimal online solution for the above problem with the causal
ESI, i.e., for any given $t$, only $E_i$'s with $t_i\leq t$ are known
at the transmitter, can be numerically solved by the technique of
dynamic programming similar to \cite{Zhang}. However, such a
solution is of high computational complexity due to ``the curse of
dimensionality'' for dynamic programming. In addition, the resulting
solution will not provide any insight to the structure of the
optimal power allocation for an energy harvesting transmitter.
Therefore, in this paper, we take an alternative approach by first
solving the off-line optimization for (\ref{eq105}),
assuming that all the energy arrival time $t_i$'s and amount $E_i$'s
are {\it a priori} known at the transmitter in each block
transmission, and then based on the structure of the off-line
optimal solution, devising online algorithms for the practical setup
with only causal ESI known at the transmitter.

For the off-line optimization of (\ref{eq105}), it is easy
to see that the objective function is concave; however, the
constraint is non-convex in general since $P_{\rm{total}}(t)$ in
(\ref{eq102}) is a concave function of $P(t)$ if $\alpha>0$. As a
result, the problem is in general non-convex and thus cannot be
solved by standard convex optimization techniques. In the next
section, we will propose an efficient solution for this problem by
exploiting its special structure.

\begin{remark}
It is worth noting that for the off-line optimization,
(\ref{eq105}) can be shown to be convex if $\alpha=0$. In this case,
similar problems to (\ref{eq105}) have been studied in the
literature \cite{Yang,Zhang}, in which the throughput-optimal power
allocation $P(t)$ was shown to follow a non-decreasing
piecewise-constant (staircase) function over $t$. This power
allocation can be interpreted as maximizing the SE of the
point-to-point AWGN channel subject to the new energy harvesting
power constraint. As will be shown later in this paper, the
non-ideal circuit power with $\alpha>0$  will change the optimal
power allocation for this problem considerably.
\end{remark}

\section{Off-Line Optimization}\label{sec:offline}

In this section, we solve the off-line optimization problem in
(\ref{eq105}) with the non-ideal circuit power, i.e., $\alpha>0$.

\subsection{Reformulated Problem}

First, we give the following lemma.{\footnote{We
thank the anonymous reviewer who brought our attention to
\cite{Bai2011}, in which an alterative proof for Lemma
\ref{Lemma:1} is given based on a calculus approach.}}

\begin{lemma}\label{Lemma:1}
During any $i$th epoch $(t_{i-1}, t_{i}]$, $i=1,\ldots,N$, the
optimal solution for (\ref{eq105}) is given by $P(t)=P_i>0$ for the
portion of time ${\mathcal T}_i^{\rm{on}} \subseteq (t_{i-1},
t_{i}]$, and $P(t)=0$ for the remaining time ${\mathcal
T}_i^{\rm{off}} \subseteq (t_{i-1}, t_{i}]$, where ${\mathcal
T}_i^{\rm{on}} \cap {\mathcal T}_i^{\rm{off}} = \phi$ and ${\mathcal
T}_i^{\rm{on}} \cup {\mathcal T}_i^{\rm{off}} = (t_{i-1}, t_{i}]$.
\end{lemma}

\begin{proof}
See Appendix \ref{appendix:proof Lemma 1}.
\end{proof}

According to Lemma \ref{Lemma:1} and by denoting the constant
transmit power $P_i>0$ for the ``on'' period with length $0 \le
l_i^{{\rm{on}}} \le L_i$ in the $i$th epoch,  (\ref{eq105})
can be reformulated as
\begin{align} \label{eq4}
\mathop {\max }\limits_{\{P_i\},~\{l_i^{{\rm{on}}}\}} ~& \sum \limits _{i=1}^{N}l_i^{{\rm{on}}}R(P_i) \nonumber \\
{\rm{s.t.}} ~ & P_i>0, ~ i=1,\ldots, N \nonumber \\
& 0\le l_i^{{\rm{on}}} \le L_i, ~ i=1,\ldots, N \nonumber \\
& \sum \limits_{j=1}^{i}{(P_j + \alpha)l_j^{{\rm{on}}}} \le \sum
\limits _{j=0}^{i-1} E_j, ~ i=1,\ldots, N.
\end{align}
However, the above problem is still non-convex due to the coupling
between $P_i$'s and $l_i^{\rm{on}}$'s. In the following, we first
solve this problem for the special case of $N=1$ and then generalize
the solution to the case with $N\geq1$.

\subsection{Single-Epoch Case with $N=1$}

In the single-epoch case with $N=1$, the problem in (\ref{eq4}) is
reduced to
\begin{align} \label{eq105oneepoch}
\mathop{\max }\limits_{l_1^{\rm{on}},P_1} ~&
{l_1^{\rm{on}}}R(P_1) \nonumber \\
{\rm s.t.} ~& P_1>0 \nonumber \\
& 0 \le {l_1^{\rm{on}}} \le T \nonumber \\
& {l_1^{\rm{on}}}(P_1+\alpha) \le E_0.
\end{align}
The solution of the above problem is given in the following
proposition.

\begin{proposition}\label{Proposition:1}
The optimal solution $P_1^*$ and $l_1^{\rm{on}*}$ for
(\ref{eq105oneepoch}) is expressed as
\begin{equation} \label{SinEpoch:2}
\displaystyle P_1^* = \max\left(P_{ee},\frac{E_0}{T}-\alpha\right)
\end{equation}
\begin{equation} \label{SinEpoch:3}
\displaystyle l_1^{\rm{on}*} = \frac{E_0}{P_1^* + \alpha}
\end{equation}
where $P_{ee}$ is given by
\begin{equation} \label{optimalEE}
\begin{array}{l}\displaystyle
P_{ee}=\mathop{\arg\max }\limits_{P_1> 0} \frac{R(P_1)}{P_1+\alpha}.
\end{array}
\end{equation}
\end{proposition}
\begin{proof}
See Appendix \ref{appendix:proof Proposition 1}.
\end{proof}

It is worth noting that $P_{ee}$ given in (\ref{optimalEE}) is the
optimal power allocation that maximizes the EE of the AWGN channel
under the non-ideal circuit power model as shown in \cite{Miao}.
From Proposition \ref{Proposition:1}, it follows that if $P_{ee}
> \frac{E_0}{T}-\alpha$, we have $P_1^* = P_{ee}$,
$l_1^{\rm{on}*} < T$  and $l_1^{\rm{off}*} = T- l_1^{\rm{on}*}
> 0$, which corresponds to an on-off transmission. However, if $P_{ee}
\leq \frac{E_0}{T}-\alpha$, we have $P_1^* =\frac{E_0}{T}-\alpha$,
$l_1^{\rm{on}*} =T$ and $l_1^{\rm{off}*} = 0$, which corresponds to
a continuous transmission. We will see in the next subsection that
the EE-maximizing power allocation $P_{ee}$ plays an important role
in the general case with $N\geq 1$. Also note that the right-hand
side (RHS) of (\ref{optimalEE}) is a quasi-concave function of $P_1$
since it is concave-over-linear \cite{Boydbook}; thus, $P_{ee}$ can
be efficiently obtained by a simple bisection search
\cite{Boydbook}.

\subsection{Multi-Epoch Case with $N\geq 1$}

Inspired by the solution for the single-epoch case, we derive the
optimal solution for (\ref{eq4}) in the general case with
$N\geq 1$, as given by the following theorem.

\begin{theorem}\label{theorem:1}
The optimal solution of (\ref{eq4}), denoted by $[P_1^*,
\ldots, P_N^*]$ and $[l_1^{{\rm{on}}*},\ldots,l_N^{{\rm{on}}*}]$, is
obtained as follows. Denoting
\begin{align}
i_{ee,0} &= 0,\nonumber\\
i_{ee,j} &= \min \bigg\{i \bigg| \frac{\sum \nolimits
_{k=i_{ee,j-1}}^{i-1} E_k}{\sum \nolimits _{k=i_{ee,j-1}+1}^{i}L_k} - \alpha \le P_{ee},\bigg.\nonumber\\
& ~~~~~~~~~~~~~~\bigg. i=i_{ee,j-1}+1,\ldots,N \bigg \}, j\ge 1,\nonumber\\
J &= \arg \max\limits_{j\ge 0}i_{ee,j},\nonumber\\
i_{ee} &= i_{ee,J}, \label{SinEpoch:12}
\end{align}
the optimal transmit power for epochs $1,\ldots,i_{ee}$ is given by
\begin{equation} \label{optimal01}
\begin{array}{l}
\displaystyle P_i^* = P_{ee}, ~ i = 1,\ldots, i_{ee}
\end{array}
\end{equation}
and the optimal on-period ${l_i^{{\rm{on}}*}}, i = 1,\ldots,i_{ee}$,
is any set of non-negative values satisfying
\begin{equation} \label{optimal02}
\begin{array}{l}
\displaystyle \left( {{P_{ee}} + \alpha } \right)\sum\nolimits_{i =
1}^{{i_{ee}}} {l_i^{{\rm{on}}*}}  = \sum\nolimits_{i = 0}^{{i_{ee}}
- 1} {{E_i}}
\end{array}
\end{equation}
\begin{equation} \label{optimal03}
\begin{array}{l}
\displaystyle \left( {{P_{ee}} + \alpha } \right)\sum\nolimits_{i =
1}^{j} {l_i^{{\rm{on}}*}}  \leq \sum\nolimits_{i = 0}^{j - 1}
{{E_i}}, ~ j=1,\ldots,i_{ee}.
\end{array}
\end{equation}
Moreover, for epochs $i_{ee} + 1,\ldots,N$, the optimal solution is
given by
\begin{equation} \label{optimal104}
\begin{array}{l}
{l_i^{{\rm{on}}*}} = L_i, ~ i = i_{ee}+1,\ldots, N
\end{array}
\end{equation}
\begin{equation} \label{optimal05}
\begin{array}{l}
\displaystyle P_i^* = \frac{\sum \nolimits _{k={n_{i-1}}}^{n_i-1}
E_k}{\sum \nolimits _{k={n_{i-1}}+1}^{n_i}L_k} - \alpha, ~ i =
i_{ee}+1,\ldots, N
\end{array}
\end{equation}
where
\begin{align} \label{optimal06}
 n_{i_{ee}} & = i_{ee},\nonumber\\
 n_i & = \arg \mathop {\min }\limits_{j: ~{n_{i-1}+1}
\leq j \leq N} \left\{ {\frac{{\sum\nolimits_{k = {n_{i-1}} }^{j -
1} {{E_k}} }}{{\sum\nolimits_{k = {n_{i-1}+1}}^j {{L_k}} }} - \alpha
} \right\}, \nonumber\\&~~~~~~~~~~~~~~~~~~~~~~~~~~~i = i_{ee}+1,\ldots, N.
\end{align}
\end{theorem}

\begin{proof}
See Appendix \ref{appendix:proof theorem 1}.
\end{proof}

It is interesting to take note that the optimal transmission policy
given in Theorem \ref{theorem:1} has a {\it two-phase} structure,
which is explained as follows in more details.

\begin{itemize}
\item $0<t \leq t_{i_{ee}}$: In the first phase, the optimal
transmission is an on-off one with a constant power $P_{ee}$ for all
the on-periods. Note that $P_{ee}$ is the EE-maximizing power
allocation given in (\ref{optimalEE}). Also note that the optimal
on-periods ${l_i^{{\rm{on}}*}}, i = 1,\ldots,i_{ee}$, may not be
unique provided that they satisfy the conditions given in
(\ref{optimal02}) and (\ref{optimal03}). Without loss of generality,
we assume that in each epoch, the transmitter chooses to be on
initially with power $P_{ee}$ provided that its stored energy is not
used up, i.e., ${l_1^{{\rm{on}}*}}=\min(E_0/(P_{ee}+\alpha),L_1)$,
${l_2^{{\rm{on}}*}}=\min\left((E_1+E_0-(P_{ee}+\alpha)l_1^{{\rm{on}}*})/(P_{ee}+\alpha),L_2\right)$,
and so on.

\item $t_{i_{ee}}<t\leq T$: In the second phase, a continuous
transmission is optimal. Since ${l_i^{{\rm{on}}*}} = L_i, ~ i =
i_{ee}+1,\ldots, N$, the problem in (\ref{eq4}) for $i = i_{ee}+1,\ldots,
N$, is reduced to
\begin{align}
\mathop {\max }\limits_{\{P_i\}} ~& \sum \limits _{i=i_{ee}+1}^{N}L_iR(P_i) \nonumber \\
{\rm{s.t.}} ~ & P_i>0, ~ i=i_{ee}+1,\ldots, N \nonumber \\
& \sum \limits_{j=i_{ee}+1}^{i}L_jP_j \le \sum \limits
_{j=i_{ee}}^{i-1} E_j- \sum \limits_{j=i_{ee}+1}^{i}L_j\alpha, \nonumber \\&~~~~~~~~~~~~~~~~~~~~~~~~~
i=i_{ee}+1,\ldots, N.
\end{align}
The optimal solution for the above problem has been shown in
\cite{Yang,Zhang} to follow a non-decreasing piecewise-constant
(staircase) function, which is given in (\ref{optimal05}). It is
worth noting that the staircase power allocation achieves the
maximum SE for an equivalent AWGN channel subject to a sequence of
energy harvesting power constraints (modified to take into account
the circuit power $\alpha$) for $i = i_{ee}+1,\ldots, N$.
\end{itemize}

From the above discussion, it is revealed that for the throughput
maximization in an energy harvesting transmission system subject to
the non-ideal circuit power, the optimal transmission unifies both
the EE and SE maximization policies independently developed in
\cite{Miao} and \cite{Yang,Zhang}, respectively. To summarize, one
algorithm for solving the problem in (\ref{eq4}) for the general case of
$N\geq 1$ is given in Table I.

\begin{table}[!ht]
\renewcommand{\arraystretch}{1.3}
\caption{Optimal Off-Line Policy for Single-Channel Case}
\label{table1} \centering
\begin{tabular}{|p{3.2in}|}
\hline
\textbf{Algorithm}\\
\hline 1)  Calculate $P_{ee}$ and obtain $i_{ee}$ as
{
\begin{align}
i_{ee,0} &= 0,\nonumber\\
i_{ee,j} &= \min \bigg\{i \bigg| \frac{\sum \nolimits
_{k=i_{ee,j-1}}^{i-1} E_k}{\sum \nolimits _{k=i_{ee,j-1}+1}^{i}L_k} - \alpha \le P_{ee},\bigg.\nonumber\\&\bigg.~~~~~~~~~~~~~~~~~~
i=i_{ee,j-1}+1,\ldots,N \bigg \}, j\ge 1,\nonumber\\
J &= \arg \max\limits_{j\ge 0}i_{ee,j},~~~i_{ee} = i_{ee,J}.\nonumber
\end{align} }
2) For the first $i_{ee}$ epochs, set
\[P_i^* = P_{ee}, i=1, \ldots, i_{ee},\]
\[l_{1}^{\rm on*} = \min\left(\frac{{E_0}}{P_{ee}+\alpha},L_1\right)\]
\[l_{i}^{\rm on*} = \min\left(\frac{\sum\nolimits_{j=1}^{i-1}{E_j}}{P_{ee}+\alpha} - \sum\nolimits_{j=1}^{i-1}{l_{j}^{\rm on*}},L_i\right), i=2, \ldots, i_{ee}.\]

3) If $i_{ee} = N$, algorithm ends; otherwise, set
\[l_{i}^{\rm on*} = L_i, i=i_{ee}+1, \ldots, N.\]

4) Reset
\[ \begin{array}{l}
\displaystyle E_0' \gets 0, T' \gets T- \sum \nolimits_{j=1}^{i_{ee}} L_j, N' \gets N -i_{ee},\\
L_j' \gets L_{j+i_{ee}}, E_{j}' \gets E_{j+i_{ee}}, j=1,\ldots,N'.\\
\end{array}
\]
5) Determine
\[i_{\min} = \arg \mathop {\min }\limits_{j} \left\{ {\frac{{\sum\nolimits_{k = 0 }^{j - 1} {{E_k'}} }}{{\sum\nolimits_{k = 0}^j {{L_k'}} }} - \alpha } \right\}\]
\[P_t = \left\{ {\frac{{\sum\nolimits_{k = 0 }^{i_{\min} - 1} {{E_k'}} }}{{\sum\nolimits_{k = 0}^{i_{\min}} {{L_k'}} }} - \alpha } \right\},\]
and set transmit power as $P_t$ in the next $i_{\min}$ epochs.\\

6) If $i_{\min} = N'$, algorithm ends; otherwise, reset the
parameters as follows, and go to 5).
\[ \begin{array}{l}
\displaystyle E_0' \gets 0, T' \gets T'- \sum \nolimits_{j=1}^{i_{\min}} L_j', N' \gets N' -i_{\min},\\
L_j' \gets L_{j+i_{\min}}', E_{j}' \gets E_{j+i_{\min}}', j=1,\ldots,N'.\\
\end{array}
\]
\\
 \hline
\end{tabular}
\end{table}

\begin{remark}\label{remark:3.1}
We discuss some implementation issues on energy harvesting communication systems with the optimal transmit power allocation given in Theorem \ref{theorem:1}. It is worth noting that in practical wireless systems, the duration of a communication block is usually on the order of
millisecond, while the energy harvesting process evolves at a much slower speed, e.g., solar and wind power typically remains constant over windows of seconds. As a result, each epoch between any two consecutive energy arrivals in our model (during which the optimal power policy in Theorem \ref{theorem:1} allocates a constant power) can be assumed to be sufficiently long, thus containing many communication blocks. In each communication block, pilot signals can be transmitted to help estimate the signal power at the receiver, which may change from one epoch to another due to transmit power adaptation; thus, the transmission rate in (\ref{sys1}) is practically achievable with AMC at the transmitter and coherent detection at the receiver.
\end{remark}

\begin{example}\label{example:1}
To illustrate the optimal two-phase transmission given in Theorem
\ref{theorem:1}, we consider an example of a band-limited AWGN channel
with bandwidth $W = 1$MHz and the noise power spectral density $N_0
= 10^{-16}$Watts-per-Hz (W/Hz). We assume that the attenuation power
loss from the transmitter to the receiver is $h=-$80dB. Considering
the channel capacity with $\Gamma=1$, we thus have $R(P) =
W\log_2(1+\frac{Ph}{\Gamma N_0W}) = \log_2(1+100P)$Mbps. It is assumed that
the energy arrives at time instants $[0, 4, 6, 11, 14, 16, 18]$sec,
and the corresponding energy values are [0.5, 0.5, 0.5, 1, 0.5,
0.75, 0.5]Joule, as shown in Fig.\ref{fig:offline}. It is also
assumed that $T=20$secs and the circuit power is $\alpha = $115.9mW.
Under this setup, we compute $P_{ee} = 79.2$mW. In Fig.
\ref{fig:offline}, we compare the optimal allocation of the total
consumed transmitter power by the algorithm in Table I with that
obtained by the algorithm given in \cite{Yang,Zhang}. Note that the
algorithm in \cite{Yang} or \cite{Zhang} solves (\ref{eq4}) in the special
case with the ideal circuit power $\alpha=0$. Here, we apply this
algorithm to obtain a suboptimal power allocation with $\alpha>0$,
by assuming that the transmitter is always on, i.e.,
${l_i^{{\rm{on}}}} = L_i, ~ i = 1,\ldots, N$. As observed in Fig.
\ref{fig:offline}(a), the optimal power allocation has a two-phase
structure, i.e., an on-off transmission with transmit power $P_{ee}$
when $0< t\le t_3$ followed by a continuous transmission with
non-decreasing staircase power allocation when $t_3< t\le T$, which
is in accordance with Theorem \ref{theorem:1}. In contrast, as
observed in Fig. \ref{fig:offline}(b), the suboptimal power
allocation by the algorithm in \cite{Yang} or \cite{Zhang} with the transmitter
always on results in a continuous transmission with non-decreasing
staircase power allocation during the entire block i.e. $0< t\leq T$.
In addition, it can be shown that the proposed optimal solution
achieves the total throughput 63.14Mbits, while the suboptimal
solution achieves only 55.80Mbits, over $T=20$secs.
\end{example}

\begin{figure}
\centering
 \epsfxsize=1\linewidth
    \includegraphics[width=8cm]{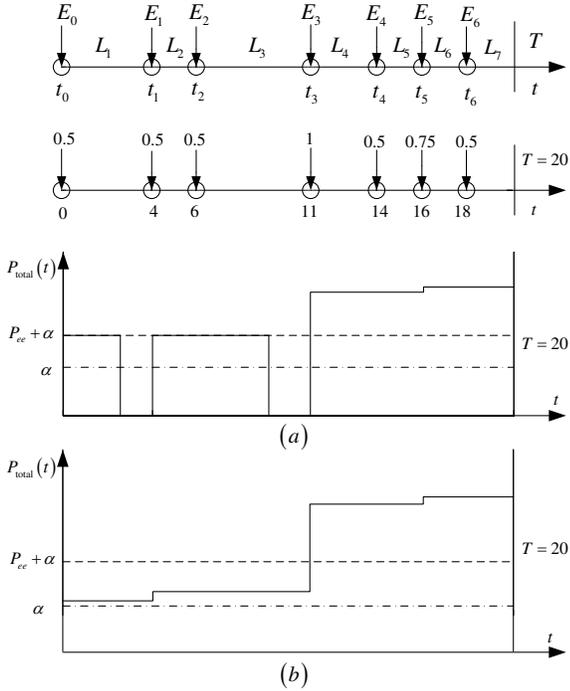}
\caption{Power allocation by off-line policies: (a) the optimal
off-line policy; and (b) the off-line policy in \cite{Yang,Zhang}.}
\label{fig:offline}
\end{figure}

\section{Multi-Channel Optimization}\label{sec:multi-channel}

In this section, we extend the optimal off-line power allocation for
the single-channel case to the general case with multiple parallel
AWGN channels subject to a total energy harvesting power constraint
and the non-ideal circuit power consumption at the transmitter. The
multi-channel setup is applicable when the communication channel is
decomposable into orthogonal channels by joint transmitter and
receiver signal processing such as OFDM (orthogonal frequency
division multiplexing) and/or MIMO (multiple-input multiple-output).

Without loss of generality, we assume a power vector ${\bf Q}(t) =
[Q_1(t),\ldots, Q_K(t)]\succeq 0, t\in(0,T]$, with each element
denoting the power allocation over time in one of a total $K$
parallel AWGN channels, where ${\bf Q}(t) \succeq 0$ denotes that ${\bf Q}(t)$
is elementwise no smaller than zero.
We also assume a sum-throughput over the $K$
channels, denoted by $C(t) = R({\bf Q}(t))$, which satisfies
\begin{enumerate}
\item $R({\bf Q}(t)) \ge 0, \forall {\bf Q}(t)\succeq 0$, and $R({\bf{0}})=0$;
\item $R({\bf Q}(t))$ is a strictly joint concave function over ${\bf
Q}(t)\succeq 0$;
\item $R({\bf Q}(t))$ is a monotonically increasing function with respect to each argument in ${\bf
Q}(t)$ i.e. $Q_k(t)\geq 0, k=1,\ldots,K$.
\end{enumerate}
An example of the above multi-channel sum-throughput is the sum-rate
over $K$ parallel AWGN channels achieved by joint AMC, which is given by
\begin{equation} \label{MultiChannel:1}
\begin{array}{l}
\displaystyle R({\bf Q}(t)) = W \sum \limits_{k=1}^K
\log_2\left(1+\frac{h_k Q_k(t)}{\Gamma WN_0}\right),
\end{array}
\end{equation}
where $h_k\geq 0$ denotes the channel power gain of the $k$th
channel. Similar to (\ref{eq102}) in the single-channel case, by
taking into account the non-ideal circuit power $\alpha$, the total
power consumed at the transmitter for the multi-channel case is
modeled by
\begin{align} \label{Gene:1}
Q_{\rm{total}}(t) = \left\{ \begin{array}{ll}
\sum \limits_{k=1}^{K} Q_k(t) + \alpha, & \sum \limits _{k=1}^{K} Q_k(t) > 0 \\
0 & \sum \limits _{k=1}^{K} Q_k(t) = 0. \end{array} \right.
\end{align}
Then the throughput maximization problem over a finite horizon $T$
in the multi-channel case is formulated as
\begin{align} \label{Gene:2}
\mathop {\max }\limits_{{\bf{Q}}(t)\succeq 0} ~&
\int_0^T R({\bf{Q}}(t)) {\rm{d}}t \nonumber \\
{\rm s.t.} ~& \int_0^{t_i} {{Q_{{\rm{total}}}}\left( t \right)}
{\rm{d}}t \le \sum\limits_{j = 0}^{i - 1} {{E_j}}, ~i=1,\ldots,N.
\end{align}
Similar to (\ref{eq105}), the above problem is non-convex
with $\alpha>0$ and thus cannot be solved by standard convex
optimization techniques. In the following, we will apply the principle of
{\it nested optimization} to convert this problem with
multi-dimensional (vector) power optimization to an equivalent
problem with only one-dimensional (scalar) power optimization, which
is then optimally solvable by the algorithm in Table I for the
single-channel case.

To apply the nested optimization, we first introduce an auxiliary
variable $P(t) = \sum \limits_{k=1}^K Q_k(t)$, and rewrite the
objective function of (\ref{Gene:2}) equivalently as
\begin{align} \label{Gene:2:object:1}
& \mathop {\max }\limits_{ P \left( t \right)\geq 0} \max
\limits_{{\bf{Q}}(t):{Q}_k(t)\geq 0, \forall k, \sum \limits_{k=1}^K
Q_k(t) \leq P(t)} \int_0^T  R({\bf{Q}}(t)) {\rm{d}}t \nonumber \\
 =& \mathop {\max }\limits_{ P\left( t \right)\geq 0} \int_0^T \max
\limits_{{\bf{Q}}(t):{Q}_k(t)\geq 0, \forall k, \sum \limits_{k=1}^K
Q_k(t) \leq P(t)} R({\bf{Q}}(t)) {\rm{d}}t.
\end{align}
Define the auxiliary function
\begin{equation} \label{Gene:3}
\begin{array}{l}
\displaystyle \bar  R(P(t)) = {\mathop {\max }\limits_{{\bf
Q}(t):{Q}_k(t)\geq 0, \forall k,  \sum \limits_{k=1}^K Q_k(t) \leq
P(t)}} R({\bf{Q}}(t))
 \end{array}
\end{equation}
for which it can be easily verified that the maximum is attained
when  $\sum \limits_{k=1}^K Q_k(t) = P(t)$. Thus, without loss of
generality, we can rewrite (\ref{Gene:2}) equivalently as
\begin{align} \label{Gene:4}
\mathop {\max }\limits_{ P\left( t \right)\geq 0} ~&
\int_0^T  \bar R( P(t)) {\rm{d}}t \nonumber \\
{\rm s.t.} ~& \displaystyle\int_0^{t_i} {{P_{{\rm{total}}}}\left( t
\right)} {\rm{d}}t \le \sum\limits_{j = 0}^{i - 1} {{E_j}},
~i=1,\ldots,N
 \end{align}
where \begin{equation} \label{gene06}\begin{array}{l}{
P_{{\rm{total}}}}(t) = \left\{ {\begin{array}{*{20}{c}}
\displaystyle {  P(t) + {\alpha},}\\
{{0},}
\end{array}} \right.\begin{array}{*{20}{l}}
{  P(t) > 0}\\
{  P(t) = 0}.
\end{array}\end{array}
\end{equation}
Thus, the original problem with vector power optimization is
converted by the nested optimization to an equivalent problem with
only scalar power optimization. Thereby, we can first solve
(\ref{Gene:4}) to get the optimal solution of $P(t)$, and then with
the obtained $P(t)$ solve (\ref{Gene:3}) to find the optimal
solution of ${\bf Q}(t)$ for (\ref{Gene:2}). Since the problem in
(\ref{Gene:3}) is a convex optimization problem, it can be solved by
standard techniques e.g. the Lagrange duality method \cite{Boydbook}
(in the special case of the sum-rate given in
(\ref{MultiChannel:1}), the optimal solution can be obtained by the
well-known ``water-filling'' algorithm \cite{Boydbook}).

In order to solve (\ref{Gene:4}), we first give the
following proposition.
\begin{proposition}\label{proposition:multichannel}
The function $\bar R(P(t))$ satisfies the following properties:
\begin{enumerate}
\item $\bar R(P(t)) \ge 0, \forall P(t)\geq 0$, and $\bar{R}(0)=0$;
\item $\bar R(P(t))$ is a strictly concave function over $P(t)\geq 0$;
\item $\bar R(P(t))$ is a monotonically increasing function over $P(t)\geq 0$.
\end{enumerate}
\end{proposition}
\begin{proof}
See Appendix \ref{appendix:proof propostion multichannel}.
\end{proof}

Since $\bar R(P(t))$ satisfies the same conditions as $R(P(t))$ for
the single-channel case, it follows that (\ref{Gene:4}) can
be similarly solved by the algorithm in Table I, with one minor
modification: the EE-maximizing power allocation in the
multi-channel case needs to be obtained as
\begin{align} \label{Gene:8}
P_{ee}&= \arg \mathop {\max }\limits_{{P}> 0} \frac{\bar
R(P)}{P+\alpha}\nonumber\\& = \arg \mathop {\max }\limits_{ {P}> 0}
\frac{\mathop {\max }\limits_{{\bf{Q}}:Q_k \ge 0, \forall k,
\sum\limits_{k = 1}^K {{Q_k}} \le P} R({\bf{Q}})}{P+\alpha}.
\end{align}
Since the maximum in the above problem is attained by
$\sum\limits_{k = 1}^K {{Q_k}}=P$, the optimal solution of ${\bf Q}$
is obtained as
\begin{equation} \label{Gene:10}
\begin{array}{l}\displaystyle
{\bf{Q}}^{ee}= \arg \mathop {\max }\limits_{{\bf{Q}}:Q_k \ge 0,
\forall k, \sum\limits_{k = 1}^K {{Q_k}} \leq P} \frac{
R({\bf{Q}})}{\sum\limits_{k = 1}^K {{Q_k}}+\alpha}
 \end{array}
\end{equation}
with ${\bf{Q}}^{ee}= [Q^{ee}_1, \ldots, Q^{ee}_K]$, and
\begin{equation} \label{Gene:11}\begin{array}{l}\displaystyle
P_{ee}= \sum\limits_{k = 1}^K {{Q^{ee}_k}}.
 \end{array}
\end{equation}
Since the RHS of (\ref{Gene:10}) is a quasi-concave function, this
problem is quasi-convex and thus can be efficiently solved by the
bisection method \cite{Boydbook}. Here we omit the detail for
brevity.

To summarize, the algorithm for solving (\ref{Gene:2}) for
the multi-channel case is given in Table II.

\begin{table}[!ht]
\renewcommand{\arraystretch}{1.3}
\caption{Optimal Off-Line Policy for Multi-Channel Case}
\label{table3} \centering
\begin{tabular}{|p{3.2in}|}
\hline
\textbf{Algorithm}\\
\hline
1) Obtain $P_{ee}$ by solving (\ref{Gene:10}) and (\ref{Gene:11}); apply the algorithm in Table \ref{table1} to obtain the solution $P^*(t)$ for (\ref{Gene:4}).\\

2) With the obtained $P^*(t)$, solve  (\ref{Gene:3})  to
obtain the solution ${\bf{Q}}^*(t)$ for  (\ref{Gene:2}). \\
\hline
\end{tabular}
\end{table}

\section{Online Algorithm}\label{sec:online}

In the previous two sections, we have studied the optimal off-line
policies assuming the non-causal ESI at the transmitter, which
provide the throughput upper bound for all online policies. In this
section, we will address the practical online case with only the
causal (past and present) ESI assumed to be known at the
transmitter. In particular, we will propose an online policy based
on the structure of the optimal off-line policy obtained previously
in Section \ref{sec:offline}. Due to the space limitation, we will
only consider the single-channel case for the study of online
algorithms, while similar results can be obtained for the general
multi-channel case, based on the optimal off-line policy given in
Section \ref{sec:multi-channel}.

\subsection{Proposed Online Algorithm}

For the purpose of exposition, we assume that the harvested energy
is modeled by a compound Poisson process, where the number of energy
arrivals over a horizon $T$ follows a Poisson distribution with mean
$\lambda_eT$ and the energy amount in each arrival is independent
and identically (i.i.d.) distributed with mean $\bar{E}$. It is
assumed that $\lambda_e$ and $\bar{E}$ are known at the transmitter.

We propose an online power allocation algorithm based on the
structure of the optimal off-line solution revealed in Theorem
\ref{theorem:1}. Specifically, considering the start time of each
block, from Theorem \ref{theorem:1}, we obtain the closed-form
solution for the optimal off-line power allocation at $t=0$ in the
following proposition.
\begin{proposition}\label{proposition:online}
Suppose there are $N-1$ energy arrivals in $(0,T)$ with $N\geq 1$,
the optimal off-line power allocation solution for
(\ref{eq4}) at $t=0$ is given by
\begin{equation} \label{Online1}
\begin{array}{l}
\displaystyle {P^*}(0) =
\max\left(\min_{i=1,\ldots,N}\left(\frac{\sum
\nolimits_{k=0}^{i-1}E_k}{\sum \nolimits_{k=1}^{i}L_k} -
\alpha\right), P_{ee}\right).
\end{array}
\end{equation}
\end{proposition}
\begin{proof}
See Appendix \ref{appendix:proof propostion online}.
\end{proof}
Note that in (\ref{Online1}), $E_0$ is available at the transmitter
at $t=0$, while $N$, $E_i, i=1,\ldots,N-1$, and $L_i, i=1,\ldots,N$,
are all unknown at the transmitter due to the causal ESI. As a
result, we cannot compute ${P^*}(0)$ in (\ref{Online1}) at $t=0$ for
the online policy. Nevertheless, we can approximate the expression
of ${P^*}(0)$ based on the statistical knowledge of the energy
arrival process, i.e., $\lambda_e$ and $\bar{E}$, as follows.

Denote \[\displaystyle \frac{{\sum \nolimits_{k=0}^{i-1}E_k}}{{\sum
\nolimits_{k=1}^{i}L_k}} = \frac{E_0 + {\sum
\nolimits_{k=1}^{i-1}E_k}}{{\sum \nolimits_{k=1}^{i}L_k}} =
\frac{E_0}{{\sum \nolimits_{k=1}^{i}L_k}} + \frac{{\sum
\nolimits_{k=1}^{i-1}E_k}}{{\sum \nolimits_{k=1}^{i}L_k}}.\]
For any $i\le N$, ${\sum \nolimits_{k=1}^{i-1}E_k}$ is the total energy
harvested during $(0,t_i)$ and ${\sum \nolimits_{k=1}^{i}L_k}=t_i$. We
thus have
\begin{equation} \label{Online2}
\frac{{\sum \nolimits_{k=1}^{i-1}E_k}}{{\sum
\nolimits_{k=1}^{i}L_k}} \approx
\frac{\lambda_et_i\bar{E}}{t_i}=\lambda_e\bar{E}, \ \forall 1< i\le N,
\end{equation}
where the approximation becomes exact when $t_i\rightarrow \infty$.

Using (\ref{Online2}), we can approximate
$\displaystyle \frac{{\sum \nolimits_{k=0}^{i-1}E_k}}{{\sum
\nolimits_{k=1}^{i}L_k}}$ as
$\displaystyle
\frac{E_0}{{\sum \nolimits_{k=1}^{i}L_k}} + \lambda_e\bar{E},$
and \begin{align*}
&\min_{i=1,\ldots,N}\left(\frac{\sum
\nolimits_{k=0}^{i-1}E_k}{\sum \nolimits_{k=1}^{i}L_k} -
\alpha\right) \\ \approx &\min_{i=1,\ldots,N}\left(\frac{E_0}
{{\sum \nolimits_{k=1}^{i}L_k}} + \lambda_e\bar{E} -
\alpha\right) = \frac{E_0}{T}
+ \lambda_e \bar{E} - \alpha,
\end{align*}
and then obtain
\begin{equation} \label{Online:4}
\begin{array}{l}
\displaystyle {P^*}(0) \approx \displaystyle \max\left(\frac{E_0}{T}
+ \lambda_e \bar{E} - \alpha, P_{ee}\right).
\end{array}
\end{equation}
Since $E_0$, $\lambda_e$ and $\bar{E}$ are all known at the
transmitter at $t=0$, (\ref{Online:4}) can be computed in real time.

For any $0<t<T$, by denoting the stored energy as $E_s(t)$ with
$E_s(0) = E_0$, we can view the online throughput maximization at
time $t$ over the remaining time $T-t$ to have an initial stored
energy $E_0=E_s(t)$. Therefore, by replacing $T$ and $E_0$ in
(\ref{Online:4}) as $T-t$ and $E_s(t)$, respectively, we obtain the
following online transmit power allocation policy:
\begin{align}\label{Online:5}
{P}_{\rm online}(t) = \left\{ \begin{array}{ll}
\max\left(\frac{E_s(t)}{T-t} + \lambda_e\bar{E} - \alpha,
P_{ee}\right), & E_s(t)>0 \\ 0, & E_s(t)=0 \end{array} \right.
\end{align}
for any $t\in[0,T)$.

The online policy in (\ref{Online:5}) provides some useful insights.
Note that $\frac{E_s(t)}{T-t} + \lambda_e\bar{E} - \alpha$ can be
viewed as the ``expected'' available transmit power for the
remaining time in each block, which can be negative for some $t$ if
$\frac{E_s(t)}{T-t} + \lambda_e\bar{E} < \alpha$. Thus, if this
value is less than the EE-maximizing power allocation $P_{ee}$, the
transmitter should transmit with $P_{ee}$ to save energy; however,
if the inequality is reversed, the transmitter should transmit more
power to maximize the SE. Moreover, as compared to the optimal
off-line power allocation for the single-epoch case given in
Proposition \ref{Proposition:1}, we see that the online policy
(\ref{Online:5}) bears a similar structure, by noting that
$E_0/T-\alpha$ in (\ref{SinEpoch:2}) for the single-epoch case is
also the available transmit power for the remaining time in each
block. Last, it is worth remarking that the online power
allocation policy in (\ref{Online:5}) is expressed as a function of the
continuous time for convenience; however, in practice, this policy needs
to be implemented in discrete time steps by properly quantizing
the continuous-time function. The time step needs to be carefully chosen
in implementation: On one hand, it is desirable to use smaller
step values to achieve higher quantization accuracy for energy saving,
while on the other hand, the time step needs to be sufficiently large,
i.e., at least  larger than one communication block (cf. Remark \ref{remark:3.1})
so that the receiver can have a timely estimate of any transmit power adjustment.

\begin{figure}
\centering
 \epsfxsize=1\linewidth
    \includegraphics[width=8cm]{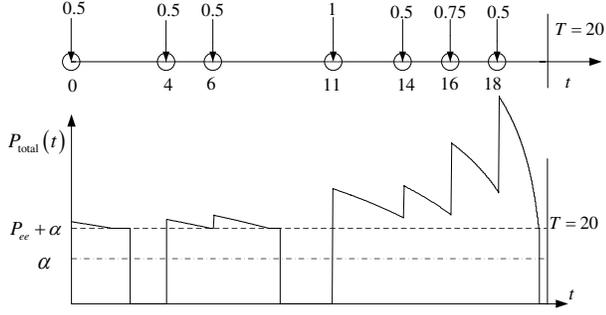}
\caption{Power allocation by the proposed online policy.}
\label{fig:online}
\end{figure}

\begin{example}
To illustrate the proposed online power allocation policy in
(\ref{Online:5}), we consider the same channel setup and harvested
energy process for the off-line case in Example \ref{example:1} (cf.
Fig. \ref{fig:offline}). In Fig. \ref{fig:online}, we show the total
transmitter power consumption by the proposed online policy assuming
that the exact average harvested power $\lambda_e\bar{E} =
(\sum_{i=1}^{N-1}E_i)/T= 187.5$mW is known at the transmitter. It is
observed that the online power allocation is no more
piecewise-constant like the optimal off-line power allocation in
Fig. \ref{fig:offline}(a). Nevertheless, it is also observed that
these two policies result in some similar power allocation patterns,
i.e., starting with an on-off power allocation followed by a
non-decreasing (in the sense of average power profile for the online
policy case) power allocation. This suggests that the proposed
online policy captures the essential features of the optimal
off-line policy. Moreover, it can be shown that the proposed online
policy achieves the total throughput 61.61Mbits over $T=20$secs,
which is only 1.53Mbits from 63.14Mbits of the optimal off-line
policy. In addition, it can be verified that the total throughput
obtained by the proposed online policy is very robust to the assumed
average harvested power value $\lambda_e\bar{E}$. For example, by
setting $\lambda_e \bar{E}$ to be 150mW or 200mW, the proposed
online policy obtains the throughput 61.38Mbits and 61.60Mbits,
respectively, which is a very small loss in either case.
\end{example}

\section{Simulation Results}\label{sec:simulation}

In the section, we compare the performance of the proposed online
policy with the performance upper bound achieved by the optimal off-line
policy under a stochastic energy harvesting setup modeled by the
compound Poisson process. The amount of energy in each energy
arrival is assumed to be independent and uniformly distributed between 0
and $2\bar{E}$. For the purpose of comparison, we also consider two
alternative heuristically designed online power allocation policies
given as follows.

\begin{itemize}
\item {\bf Energy Efficient Policy (EEP)}: In this online policy, the transmitter transmits with the EE-maximizing power allocation $P_{ee}$ given
in (\ref{optimalEE}) provided that there is a non-zero stored
energy, i.e.,
\begin{align}\label{Online:6}
{P}_{\rm EEP}(t) = \left\{ \begin{array}{ll} P_{ee}, & E_s(t)>0 \\
0, & E_s(t)=0 \end{array} \right.
\end{align}
for any $t\in[0,T)$.

\item {\bf Energy Neutralization Policy (ENP)}: This online policy transmits with a constant
power that satisfies the long-term energy consumption constraint if
there is available stored energy, i.e.,
\begin{align}\label{Online:7}
{P}_{\rm ENP}(t) = \left\{ \begin{array}{ll} \lambda_e \bar{E}  - \alpha, & E_s(t)>0 \\
0, & E_s(t)=0 \end{array} \right.
\end{align}
for any $t\in[0,T)$. Note that in the above we have assumed that
$\lambda_e \bar{E}>\alpha$.
\end{itemize}

\begin{figure}
\centering
 \epsfxsize=1\linewidth
    \includegraphics[width=8cm]{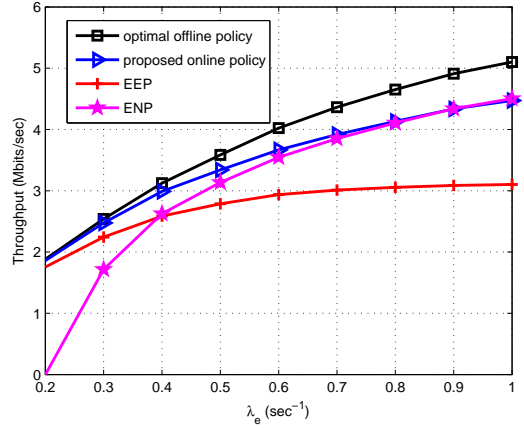}
\caption{Average throughput versus the energy arrival rate $\lambda_e$ with $\bar E = 0.5$J and $T=$20secs.} \label{fig:lambda}
\end{figure}

\begin{figure}
\centering
 \epsfxsize=1\linewidth
    \includegraphics[width=8cm]{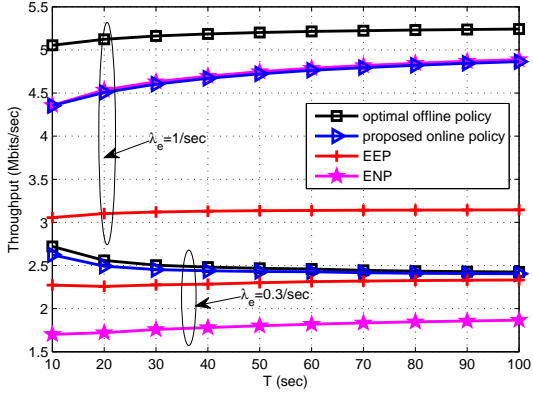}
\caption{Average throughput versus the block duration $T$ with $\bar
E = $0.5J.} \label{fig:horizon}
\end{figure}

\begin{figure}
\centering
 \epsfxsize=1\linewidth
    \includegraphics[width=8cm]{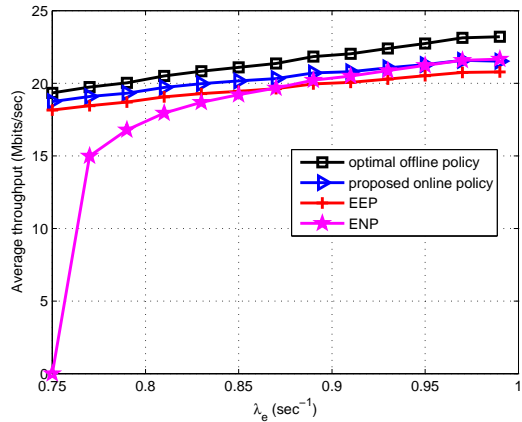}
\caption{Average throughput in a single-cell OFDMA downlink system with renewable powered BS versus the energy arrival rate $\lambda_e$ with $\bar E = 200$J and $T=$20secs.} \label{fig:OFDMA}
\end{figure}

First, we consider a single AWGN channel in
Figs. \ref{fig:lambda} and \ref{fig:horizon} with the same channel parameters
as for Example \ref{example:1}.
In Fig. \ref{fig:lambda}, we show the average throughput over
$T=20$secs versus $\lambda_e$, with $\bar{E}=0.5$Joule (J). It is
observed that when $\lambda_e$ is small, the proposed online policy
and EEP obtain similar performance as the optimal off-line policy.
Since the average harvested energy is small when $\lambda_e$ is
small, it is more likely that $P_{ee}$ is greater than both
$\frac{E_s(t)}{T-t} + \lambda_e E - \alpha$ (c.f. (\ref{Online:5})) and $\frac{\sum
\nolimits _{k=0}^{i-1} E_k}{\sum \nolimits _{k=1}^{i}L_k} - \alpha, \forall i=1,\ldots,N$ (c.f. (\ref{Online1})).
Thus, both the proposed online and off-line policies choose to
transmit with $P_{ee}$ during the ``on'' periods to save energy. However, for ENP, the
transmit power level deviates from $P_{ee}$ and thus significant
amount of energy is consumed due to the non-ideal circuit power; as
a result, the achievable throughput is almost zero. As $\lambda_e$
increases, the throughput gap between the optimal off-line policy
and all online policies enlarges, and the performance of EEP
degenerates severely. Moreover, ENP is observed to obtain a similar
performance as the proposed online policy. This is because in this
case, $\frac{E_s(t)}{T-t}$ is negligibly small as compared with
$\lambda_e E$, and as a result the proposed online policy
degenerates to ENP.

Fig. \ref{fig:horizon} shows the average throughput versus
$T$, for two different values of $\lambda_e=0.3$/sec and 1/sec, with
$\bar{E}=0.5$J. In both cases of  $\lambda_e$, the proposed online
policy is observed to perform close to the optimal off-line policy,
for all values of $T$. With small value of $\lambda_e$, i.e.,
$\lambda_e=0.3$/sec, EEP performs much better than ENP since it is
more energy efficient, while with a larger value of $\lambda_e$, i.e.,
$\lambda_e=1$/sec, the reverse becomes true, which can be similarly
explained as for Fig. \ref{fig:lambda}.

Furthermore, for the
multi-channel scenario, in Fig. \ref{fig:OFDMA} we evaluate the average
throughput of a single-cell downlink system with the base station (BS)
powered by energy harvesting. It is assumed that the BS covers a circular
area with radius 1000 meters, and serves $K$ users whose locations are generated following a spatial
homogeneous Poisson point process (HPPP) with density $10^{-6}$
users/m$^2$. Consider a simplified channel model without fading, in which
the channel power gain of each
user from the BS is determined by a pathloss model
$c_0\left(\frac{r}{r_0}\right)^{-\zeta}$, where $c_0 = -60$dB
is a constant equal to the pathloss at a reference distance $d_0 = 10$m,
and $\zeta = 3$ is the pathloss exponent. Assuming an OFDMA (orthogonal
frequency division multiplexing access) based user multiple access, a total
bandwidth $W = 5$MHz is equally allocated to the $K$ users. The
noise power spectral density at each user receiver is set as  $N_0 = -174$dBm/Hz and
$\Gamma=1$ is assumed.
The circuit power at the BS is set as $\alpha = 60$Watt.
Fig. \ref{fig:OFDMA} shows the average throughput versus $\lambda_e$
with $\bar E = 200$J and an energy scheduling period $T=$20secs.
It is observed that the proposed online policy always performs better
than EEP. With small values of $\lambda_e$, i.e., $\lambda_e\le0.8$/sec,
ENP preforms clearly worse than the proposed online algorithm, while when
$\lambda_e\ge0.95$/sec, both schemes perform similarly. This observation is  expected as
can be similarly explained for Fig.
\ref{fig:lambda}.

\section{Concluding Remarks}\label{sec:conclusion}

In this paper, we studied the throughput-optimal transmission
policies for energy harvesting wireless transmitters with the
non-ideal circuit power. We first obtained the optimal off-line
solution in the single-channel case, which is shown to have a new
two-phase transmission structure by unifying existing results on
separately maximizing energy efficiency and spectrum efficiency. We
then extended the optimal off-line solution to the general case with
multiple AWGN channels subject to a total energy harvesting power
constraint, by the technique of nested
optimization. Finally, we proposed an online algorithm based on a
closed-form off-line solution. It is shown by simulations that the
proposed online algorithm has a very close performance to the upper
bound achieved by the optimal off-line solution, and also outperforms other
heuristically designed online algorithms.

After submission of this manuscript, we become aware
of one interesting related work \cite{Devillers2012} that is worth
mentioning. In \cite{Devillers2012}, the throughput optimization
in a single-channel  energy
harvesting communication system with battery leakage is introduced.
The impact of battery leakage is very similar to that of the non-ideal
circuit power considered in this paper; as a result, the optimal
transmission policy developed in \cite{Devillers2012} is
similar to the one proposed in Section \ref{sec:offline} of this paper. One difference
is that the optimal off-line policy in \cite{Devillers2012}
is developed by decoupling the general multi-epoch problem into multiple
equivalent single-epoch subproblems, while the same optimal policy in
this paper is derived by decoupling the problem into only two subproblems (c.f. Appendix \ref{appendix:proof theorem 1}).
Compared with the solution in \cite{Devillers2012}, the off-line
policy in this paper reveals the optimal two-phase
structure in which EE and SE optimizations are unified, and thus motivates our online policies,
which are not given in \cite{Devillers2012}.

\appendices

\section{Proof of Lemma \ref{Lemma:1}} \label{appendix:proof Lemma 1}

Denote the length of the on-period ${\mathcal T}_i^{\rm{on}}$ as
$l_{i}^{\rm{on}}$ and that of the off-period ${\mathcal
T}_i^{\rm{off}}$ as $l_{i}^{\rm{off}}$, where $l_{i}^{\rm{on}} +
l_{i}^{\rm{off}} = L_i$. Without loss of generality, we only need to
consider the case with on-period ${\mathcal T}_i^{\rm{on}} = (t_{i-1},
t_{i-1}+l_{i}^{\rm{on}}]$ and off-period ${\mathcal T}_i^{\rm{off}} =
(t_{i-1}+l_{i}^{\rm{on}},t_{i}]$, since exchanging the power
allocation at two different time instants in an epoch does not
change the throughput and the energy constraint. Next, we
prove that the transmit power should be constant during the
on-period ${\mathcal T}_i^{\rm{on}}$ in an epoch by contradiction.

Suppose that the optimal allocated transmit power $\bar P(t)$, where
$\bar P(t) > 0, t \in (t_{i-1}, t_{i-1}+l_{i}^{\rm{on}}]$ is not constant.
Since $R(P(t))$ is a strictly concave function, based on Jensen's
inequality, we have
\begin{equation} \label{SinEpoch:6}
 R\left(\frac{\int_{t_{i-1}}^{t_{i-1}+l_i^{\rm{on}}} \bar P(t){\rm{d}}t}{l_i^{\rm{on}}}\right) > \int_{t_{i-1}}^{t_{i-1}+l_i^{\rm{on}}} \frac{ R( \bar P(t))}{l_i^{\rm{on}} }{\rm{d}}t,
\end{equation}
and then
\begin{align} \label{SinEpoch:7}
& \int_{t_{i-1}}^{t_{i-1}+l_i^{\rm{on}}}
R\left(\frac{\int_{t_{i-1}}^{t_{i-1}+l_i^{\rm{on}}} \bar
P(t){\rm{d}}t}{l_i^{\rm{on}}}\right){\rm{d}}t \nonumber\\= &{l_i^{\rm{on}} }
R\left(\frac{\int_{t_{i-1}}^{t_{i-1}+l_i^{\rm{on}}} \bar
P(t){\rm{d}}t}{l_i^{\rm{on}}}\right)  \nonumber\\ >&
\int_{t_{i-1}}^{t_{i-1}+l_i^{\rm{on}}} { R( \bar P(t))}{\rm{d}}t.
\end{align}
Thus, if we construct a new transmit power allocation $\hat P(t)$ as
$\hat P(t) = P_i = \frac{\int_{t_{i-1}}^{t_{i-1}+l_i^{\rm{on}}} \bar
P(t){\rm{d}}t}{l_i^{\rm{on}}}>0, t \in (t_{i-1},t_{i-1}+l_i^{\rm{on}}]$, we
can achieve a larger throughput than that achieved by $\bar P(t)$. Moreover, we
verify that $\hat P(t)$ consumes the same total energy as $\bar
P(t)$ in the $i$th epoch, i.e.,
\begin{align} \label{reformulation:10}
&\int_{t_{i-1}}^{t_{i-1}+l_i^{\rm{on}}} \left(\hat
P(t)+\alpha\right) {\rm{d}}t = l_i^{\rm{on}} (P_i+\alpha) \nonumber\\ =&
\int_{t_{i-1}}^{t_{i-1}+l_i^{\rm{on}}} \bar P(t){\rm{d}}t + l_i^{\rm{on}}
\alpha   = \int_{t_{i-1}}^{t_{i-1}+l_i^{\rm{on}}} \left(\bar
P(t)+\alpha\right){\rm{d}}t.
\end{align}
Therefore, based on   (\ref{SinEpoch:7}) and
(\ref{reformulation:10}), we conclude that $\bar P(t)$ cannot be
optimal and thus Lemma \ref{Lemma:1} is proved.

\section{Proof of Proposition \ref{Proposition:1}} \label{appendix:proof Proposition 1}

To solve (\ref{eq105oneepoch}), we note that the third
inequality constraint must be met with equality by the optimal
solution, since otherwise the throughput can be further improved by
increasing $P_1$. Thus, by substituting ${l_1^{\rm{on}}}=
\frac{E_0}{P_1+\alpha}$ into the objective function as well as the
constraint $l_1^{\rm on}\leq T$, the problem becomes equivalent to
finding
\begin{align}
\label{eq105oneepochequ} P_1^* &= \arg
\mathop { \max }\limits_{P_1>0, P_1\geq E_0/T-\alpha}
\frac{E_0}{P_1+\alpha} R(P_1) \nonumber\\&= \arg \mathop { \max }\limits_{P_1>0,
P_1 \geq E_0/T-\alpha} \frac{R(P_1)}{P_1+\alpha}.
\end{align}
Consider first the following problem with the relaxed power
constraint:
\begin{equation}
\begin{array}{l}\displaystyle
\mathop {\max }\limits_{P_1> 0} \frac{R(P_1)}{P_1+\alpha}.
\end{array}
\end{equation}
This problem has been studied in \cite{Miao}, where the globally
optimal solution is known as the EE-maximizing power allocation,
denoted by $P_{ee}$. It was also shown in \cite{Miao} that given
$\alpha>0$, $\frac{R(P_1)}{P_1+\alpha} $ is monotonically increasing
with $P_1$ if  $0\leq P_1 <P_{ee}$, and monotonically decreasing
with $P_1$ if $P_1>P_{ee}$. Thus, the solution of
(\ref{eq105oneepochequ}) is obtained as
 \begin{equation} \label{SinEpoch:8}
P_1^* = \max\left(P_{ee}, \frac{E_0}{T}-\alpha\right).
\end{equation}
Accordingly, the optimal on-period is given by
 \begin{equation} \label{SinEpoch:9}
\displaystyle l_1^{\rm{on}*} = \frac{E_0}{P_1^* + \alpha}.
\end{equation}
Proposition \ref{Proposition:1} is thus proved.

\section{Proof of Theorem \ref{theorem:1}} \label{appendix:proof theorem 1}

To prove Theorem \ref{theorem:1}, we construct the following two
sub-problems ${\mathbb{P}}_1$ and ${\mathbb{P}}_2$ for the power
allocation optimization in the first $i_{ee}$ epochs and the last
$N-i_{ee}$ epochs, respectively.
\begin{align}\label{equ:AppC:1}
{\mathbb{P}}_1: \mathop {\max }\limits_{\{P_i\},\{l_i^{{\rm{on}}}\}} ~ & \sum\limits_{i = 1}^{{i_{ee}}} {l_i^{{\rm{on}}}} R({P_i}), \nonumber \\
 {\rm{s}}{\rm{.t}}{\rm{. }} ~& 0 \le l_i^{{\rm{on}}} \le {L_i},i = 1, \ldots ,{i_{ee}}, \nonumber \\
 ~& \sum\limits_{j = 1}^i {({P_j} + \alpha )l_j^{{\rm{on}}}}  \le \sum\limits_{j = 0}^{i - 1} {{E_j}} ,i = 1, \ldots
 ,{i_{ee}}.
\end{align}
\begin{align}\label{equ:AppC:2}
{\mathbb{P}}_2:  \mathop {\max }\limits_{\{{P_i}\},\{l_i^{{\rm{on}}}\}} ~ & \sum\limits_{i = {i_{ee}} + 1}^N {l_i^{{\rm{on}}}} R({P_i}), \nonumber \\
 {\rm{s}}{\rm{.t}}{\rm{. }}~ & 0 \le l_i^{{\rm{on}}} \le {L_i},i = {i_{ee}} + 1, \ldots ,N \nonumber \\
~ & \sum\limits_{j = {i_{ee}} + 1}^i {({P_j} + \alpha
)l_j^{{\rm{on}}}}  \le \sum\limits_{j = {i_{ee}}}^{i - 1} {{E_j}} ,\nonumber \\&~~~~~~~~~~~~~~~~~~~~~~i
= {i_{ee}} + 1, \ldots ,N.
\end{align}
We will first prove that the solution given in Theorem
\ref{theorem:1} is optimal for both ${\mathbb{P}}_1$ and
${\mathbb{P}}_2$, and then prove that the optimal solutions for
${\mathbb{P}}_1$ and ${\mathbb{P}}_2$ are also optimal for
(\ref{eq4}).

First, we prove that the solution given in (\ref{optimal01}),
(\ref{optimal02}) and (\ref{optimal03}) is optimal for
${\mathbb{P}}_1$.

Consider the throughput maximization problem ${\mathbb{P}}_1$, with
the arrived energy $E_0, \ldots, E_{i_{ee}-1}$ at time $t_0, \ldots,
t_{i_{ee}-1}$ over the horizon $T_1={\sum \nolimits _{k=1}^{i_{ee}}
L_k}$. We construct an auxiliary throughput maximization problem
$\bar{\mathbb{P}}_1$ with the energy arrival ${\sum \nolimits
_{k=0}^{i_{ee}-1} E_k}, 0, \ldots, 0$ at time $t_0, \ldots,
t_{i_{ee}-1}$ over the same horizon $T_1$ as
follows.
\begin{align}\label{equ:AppC:3}
\bar{\mathbb{P}}_1: \mathop {\max }\limits_{\{{P_i}\},\{l_i^{{\rm{on}}}\}} ~ & \sum\limits_{i = 1}^{{i_{ee}}} {l_i^{{\rm{on}}}} R({P_i}), \nonumber \\
 {\rm{s}}{\rm{.t}}{\rm{. }} ~& 0 \le l_i^{{\rm{on}}} \le {L_i},i = 1, \ldots ,{i_{ee}}, \nonumber \\
 ~& \sum\limits_{j = 1}^{i_{ee}} {({P_j} + \alpha )l_j^{{\rm{on}}}}  \le \sum\limits_{j = 0}^{i_{ee} - 1} {{E_j}}.
\end{align}
It is clear that the optimal throughput of $\bar{\mathbb{P}}_1$ is
an upper bound on that of ${\mathbb{P}}_1$
since any feasible solution of ${\mathbb{P}}_1$
is also feasible for $\bar{\mathbb{P}}_1$. Note that there is no energy
arrived in $t_1, \ldots, t_{i_{ee}-1}$ for $\bar{\mathbb{P}}_1$, so
$\bar{\mathbb{P}}_1$ is indeed equivalent to
a throughput maximization problem for the single-epoch case
studied in Section III-B over a horizon $T_1$. It can be easily verified
based on (\ref{SinEpoch:12}) that $\displaystyle \frac{
\sum\nolimits_{j = 0}^{{i_{ee}} - 1} {{E_j}}}{T_1} - \alpha \le P_{ee}$;
thus, it follows after some simple manipulation that the optimal value
of $\bar{\mathbb{P}}_1$ is $\displaystyle \frac{
\sum\nolimits_{j = 0}^{{i_{ee}} - 1} {{E_j}}}{P_{ee} + \alpha} \cdot R(P_{ee})$
and is attained by $\bar P_1^* = \cdots =
\bar P_{i_{ee}}^*  = P_{ee}$ and $\displaystyle \sum_{j=1}^{i_{ee}}\bar l_j^{{\rm{on}}*} = \frac{
\sum\nolimits_{j = 0}^{{i_{ee}} - 1} {{E_j}}}{P_{ee} + \alpha}$. Meanwhile,
for ${\mathbb{P}}_1$ we can always construct a feasible solution based on (\ref{optimal01}), (\ref{optimal02}) and
(\ref{optimal03}) by setting
\begin{align}
P_j^*  & =P_{ee},\ \forall j=1,\ldots,i_{ee},\nonumber \\
l_k^{{{\rm{on}}}*} &= L_k,\ \forall k\neq i_{ee,j}, \forall j=1,\ldots,J,\nonumber \\
l_{i_{ee,j}}^{{\rm{on}}*} &= \frac{\sum_{k=i_{ee,j-1}}^{i_{ee,j}-1}E_k}{P_{ee}+\alpha}- \sum_{k=i_{ee,j-1}+1}^{i_{ee,j}}L_k,\  \forall j=1,\ldots,J.\label{equ:revision:1}
\end{align}
It can be verified from (\ref{SinEpoch:12})
that the solution satisfying (\ref{equ:revision:1}) is feasible
for ${\mathbb{P}}_1$, and attains an objective value of $\displaystyle \frac{
\sum\nolimits_{j = 0}^{{i_{ee}} - 1} {{E_j}}}{P_{ee} + \alpha} \cdot R(P_{ee})$, which is the same as the optimal value of $\bar{\mathbb{P}}_1$.
The gap between ${\mathbb{P}}_1$ and $\bar{\mathbb{P}}_1$ is thus zero,
and accordingly, the solution given in (\ref{optimal01}),
(\ref{optimal02}) and (\ref{optimal03}) is optimal for
${\mathbb{P}}_1$.

Second, we prove that the solution in (\ref{optimal104}),
(\ref{optimal05}) and (\ref{optimal06}) is optimal for Problem
${\mathbb P}_2$ over the last $N-i_{ee}$ epochs. First, we show
${l_i^{{\rm{on}}*}} = L_i, i = i_{ee}+1,\ldots, N$ by contradiction
as follows. Suppose that the optimal solution $\hat P(t)$ contains
an ``off'' period with $(\hat t^{\rm{off}}, \hat t^{\rm{off}}+\Delta
\hat t^{\rm{off}}) \subset (t_{i_{ee}},T]$, i.e., $\hat P(t) = 0, t
\in (\hat t^{\rm{off}}, \hat t^{\rm{off}}+\Delta \hat
t^{\rm{off}})$.

Due to the definition of  (\ref{SinEpoch:12}), it follows
immediately that
\begin{equation} \label{Proof101}
\begin{array}{l}\displaystyle \frac{\sum \nolimits _{k=i_{ee}}^{i-1} E_k}{\sum \nolimits _{k=i_{ee}+1}^{i}L_k} - \alpha > P_{ee}, \forall i > i_{ee}\end{array}
\end{equation}
Therefore, we can always find a time
duration with $\hat P(t) = \hat P^{\rm{on}} > P_{ee}, t \in (\hat
t^{\rm{on}}, \hat t^{\rm{on}}+\Delta \hat t^{\rm{on}}) \subset
(t_{i_{ee}},T]$, and construct a new policy $\displaystyle \bar
P(t)$ with $\displaystyle \bar P(t) = \bar P^{\rm{on}} = \frac{(\hat
P^{\rm{on}} + \alpha)\Delta \hat t^{\rm{on}}}{\Delta \hat
t^{\rm{on}} + \delta} - \alpha, t\in (\hat t^{\rm{on}}, \hat
t^{\rm{on}}+\Delta \hat t^{\rm{on}}) \cup (\hat t^{\rm{off}}, \hat
t^{\rm{off}}+\delta) $. Note that we have chosen $\delta$ to be
sufficiently small so that $(\hat t^{\rm{off}}, \hat
t^{\rm{off}}+\delta) \subseteq (\hat t^{\rm{off}}, \hat
t^{\rm{off}}+\Delta \hat t^{\rm{off}})$ and $\bar P^{\rm{on}} >
P_{ee}$. The energy consumed by $\bar P(t)$ during $(\hat
t^{\rm{on}}, \hat t^{\rm{on}}+\Delta \hat t^{\rm{on}}) \cup (\hat
t^{\rm{off}}, \hat t^{\rm{off}}+\delta)$ is $\hat E = (\bar
P^{\rm{on}} + \alpha)(\Delta \hat t^{\rm{on}}+\delta) = (\hat
P^{\rm{on}} + \alpha) \Delta \hat t^{\rm{on}}$, which is same as
the initial policy $\hat P(t)$. The throughput for the newly
constructed policy $\bar P(t)$ and initial policy $\hat P(t)$ during
$(\hat t^{\rm{on}}, \hat t^{\rm{on}}+\Delta \hat t^{\rm{on}}) \cup
(\hat t^{\rm{off}}, \hat t^{\rm{off}}+\delta) $ are
\[
\displaystyle B_1 = R(\bar P^{\rm{on}})(\Delta \hat t^{\rm{on}}+\delta) =
R(\bar P^{\rm{on}}) \frac{\hat E}{\bar P^{\rm{on}} + \alpha}
\]
and
\[
\displaystyle B_2 = R(\hat P^{\rm{on}})\Delta \hat
t^{\rm{on}} = R(\hat P^{\rm{on}}) \frac{\hat E}{\hat
P^{\rm{on}} + \alpha}
\]
respectively. Since $\hat P^{\rm{on}} > \bar P^{\rm{on}} > P_{ee}$,
and $ \frac{R(x)}{x + \alpha}$ is monotonically decreasing as a
function of $x$ when $x>P_{ee}$, we conclude that $B_1 > B_2$.
Therefore, the new policy achieves a higher throughput than the
initial policy.

Moreover, we need to check that the new policy also satisfies the
energy constraint as follows. If $\hat t^{\rm{off}} > \hat
t^{\rm{on}}+\Delta \hat t^{\rm{on}}$, or $(\hat t^{\rm{on}}, \hat
t^{\rm{on}}+\Delta \hat t^{\rm{on}}) $ and $ (\hat t^{\rm{off}},
\hat t^{\rm{off}}+\delta) $ are in the same epoch, it is evident
that the energy constraint is satisfied. If we cannot find an
interval $(\hat t^{\rm{off}}, \hat t^{\rm{off}}+\Delta \hat
t^{\rm{off}})$ latter than  or in the same epoch as $(\hat
t^{\rm{on}}, \hat t^{\rm{on}}+\Delta \hat t^{\rm{on}})$, i.e., all
the allocated power prior to $\hat t^{\rm{off}}$ is smaller than
$P_{ee}$, in this case based on (\ref{Proof101}), there must be some
energy left at time $\hat t^{\rm{off}}$. Since we selected $\delta$
to be sufficiently small, we can still guarantee that the energy
constraint is satisfied. Therefore, we prove that
${l_i^{{\rm{on}}*}} = L_i, i = i_{ee}+1,\ldots, N$.

After determining ${l_i^{{\rm{on}}*}} = L_i, i = i_{ee}+1,\ldots,
N$, we realize that Problem ${\mathbb P}_2$ for the last $N-i_{ee}$
epochs has the same structure as that studied in \cite{Yang}, and
thus \cite[Theorem 1]{Yang} is applicable here. After some change of
notation, we can prove that the solution in (\ref{optimal104}),
(\ref{optimal05}) and (\ref{optimal06}) is optimal for ${\mathbb
P}_2$ over the last $N-i_{ee}$ epochs.

Next, we show that the optimal solutions of ${\mathbb P}_1$ and
${\mathbb P}_2$ are optimal for (\ref{eq4}) to complete the
proof for Theorem \ref{theorem:1}. First, we can verify by following
the similar contradiction proof in the above that any energy
harvested during $[0,t_{i_{ee}})$ should be used up in the first
$i_{ee}$ epochs, i.e., $\sum\limits_{j = 1}^{{i_{ee}}} {({P_j^*} +
\alpha )l_j^{{\rm{on*}}}}  = \sum\limits_{j = 0}^{{i_{ee}} - 1}
{{E_j}}$. Therefore, the energy constraint for (\ref{eq4})
is equivalent to
\begin{equation} \label{Appendix2:01}
\begin{array}{l}
\sum \limits _{j=1}^{i}{(P_j + \alpha)l_j^{{\rm{on}}}} \le \sum \limits _{j=0}^{i-1} E_j, i=1,\ldots, i_{ee}\\
\sum \limits _{j=i_{ee}+1}^{i}{(P_j + \alpha)l_j^{{\rm{on}}}} \le
\sum \limits _{j=i_{ee}}^{i-1} E_j, i=i_{ee}+1,\ldots, N.
\end{array}
\end{equation}
Since the energy constraint is decoupled before and after $t_{i_{ee}}$,
solving (\ref{eq4}) is equivalent to optimizing ${P_i}$ and
$l_i^{{\rm{on}}}$ over $i = {1, \ldots ,i_{ee}}$ and $i ={i_{ee}} +
1, \ldots ,N$ separately. Thus, the optimal solutions for
${\mathbb{P}}_1$ and ${\mathbb{P}}_2$ are also optimal for
(\ref{eq4}). Theorem \ref{theorem:1} is thus proved.

\section{Proof of Proposition \ref{proposition:multichannel}} \label{appendix:proof propostion multichannel}

The first and third properties of $\bar R(P(t))$ can be directly
verified by the first and third properties of $R({\bf Q}(t))$,
respectively. Thus, to complete the proof of Proposition
\ref{proposition:multichannel}, we only need to show the second
property of $\bar R(P(t))$, i.e., it is a strictly concave function
of $P(t)$. Similar to \cite[Appendix B]{ZhangCR}, we show the proof
of this result as follows.

Since $\bar R(P(t))$ is obtained as the optimal value of
(\ref{Gene:3}), which is a convex optimization problem and satisfies
the Slater's condition \cite{Boydbook}. Thus, the duality gap for
this problem is zero. As a result, $\bar R(P(t))$ can be
equivalently obtained as the optimal value of the following min-max
optimization problem:
\begin{align} \label{equ:appendixB:1}
\bar R(P) &= \min\limits_{\mu \ge 0} \max\limits_{Q_k \geq 0}
R({\bf{Q}}) - \mu\left(\sum \limits_{k=1}^K Q_k - P\right)
\\ \label{equ:appendixB:2}
& = \min\limits_{\mu \ge 0}  R({\bf{Q}}^{(\mu)}) - \mu\sum
\limits_{k=1}^K Q_k^{(\mu)} + \mu P \\
\label{equ:appendixB:3} & = R({\bf{Q}}^{(\mu^{(P)})}) - \mu^{(P)}
\sum \limits_{k=1}^K Q_k^{(\mu^{(P)})} +\mu^{(P)} P
\end{align}
where we have removed $t$ for brevity, and in
(\ref{equ:appendixB:2}) ${\bf{Q}}^{(\mu)} =
[{Q}_1^{(\mu)},\ldots,{Q}_K^{(\mu)}]$ is the optimal solution for
the maximization problem with a given $\mu$, while in
(\ref{equ:appendixB:3}) $\mu^{(P)}$ is the optimal solution for the
minimization problem with a given $P$. Since $R({\bf{Q}}(t))$ is a
strictly joint concave function, the optimal solutions in the above
must be unique. Denote $\omega$ as any constant in $[0,1]$. Let
$\mu^{(P_1)}$, $\mu^{(P_2)}$ and $\mu^{(P_3)}$ be the optimal $\mu$
for $\bar R(P_1)$, $\bar R(P_2)$ and  $\bar R(P_3)$ with $P_3=\omega
P_1 + (1-\omega)P_2$, respectively. For $j=1,2$, we have
\begin{align} \label{equ:appendixB:4}
\bar R(P_j) &=  R({\bf{Q}}^{(\mu^{(P_j)})}) - \mu^{(P_j)} \sum
\limits_{k=1}^K Q_k^{(\mu^{(P_j)})} + \mu^{(P_j)} P_j \\
\label{equ:appendixB:5} & \le R({\bf{Q}}^{(\mu^{(P_3)})}) -
\mu^{(P_3)} \sum \limits_{k=1}^K Q_k^{(\mu^{(P_3)})} + \mu^{(P_3)}
P_j
\end{align}
where {\it strict inequality} holds for (\ref{equ:appendixB:5})
if $P_j \neq P_3$ since the optimal $\mu^{(P_j)}, j=1,2$, are unique.
Thus, we have
\begin{align} \label{equ:appendixB:6}
& \omega \bar R(P_1) + (1-\omega) \bar R(P_2)
\\ \le& R({\bf{Q}}^{(\mu^{(P_3)})}) - \mu^{(P_3)} \sum \limits_{k=1}^K Q_k^{(\mu^{(P_3)})} + \mu^{(P_3)} P_3\label{equ:appendixB:6-1}
\\ \label{equ:appendixB:7}
=& \bar R(P_3) \\ =& \bar R(\omega P_1 + (1-\omega)P_2),
\end{align}
where {\it strict inequality} holds for (\ref{equ:appendixB:6-1})
if $\omega\in(0,1)$. Therefore, $\bar R(P(t))$ is a strictly concave function over
$P(t)\geq 0$.  The proof of Proposition
\ref{proposition:multichannel} is thus completed.

\section{Proof of Proposition \ref{proposition:online}} \label{appendix:proof propostion online}

We prove this proposition by considering the following two cases.

First, consider the case when $i_{ee}$ exists, i.e.,
$i_{ee}\in\{1,\ldots,N\}$. In this case, setting $P_{ee}$ as the
transmit power at time 0 is optimal according to Table \ref{table1}.
Furthermore, we have $\frac{\sum \nolimits_{k=0}^{i_{ee,1}-1}E_k}{\sum
\nolimits_{k=1}^{i_{ee,1}}L_k} - \alpha < P_{ee}$ based on (\ref{SinEpoch:12}), and thus
\[
\displaystyle \min_{i=1,\ldots,N}\left(\frac{\sum
\nolimits_{k=0}^{i-1}E_k}{\sum \nolimits_{k=1}^{i}L_k} -
\alpha\right) \leq \frac{\sum \nolimits_{k=0}^{i_{ee,1}-1}E_k}{\sum
\nolimits_{k=1}^{i_{ee,1}}L_k} - \alpha < P_{ee}.
\]
Thus, (\ref{Online1}) is equivalent to $P^*(0) = P_{ee}$, which is
the optimal solution.

Next, consider the case when $i_{ee}$ does not exist, which implies
that $\frac{\sum \nolimits_{k=0}^{j-1}E_k}{\sum
\nolimits_{k=1}^{j}L_k} - \alpha > P_{ee}, \forall j=1,\ldots,N$.
Thus, we have
\[
\displaystyle \min_{i=1,\ldots,N}\left(\frac{\sum
\nolimits_{k=0}^{i-1}E_k}{\sum \nolimits_{k=1}^{i}L_k} -
\alpha\right) > P_{ee}
\]
and (\ref{Online1}) is equivalent to $P^*(0) = \min \limits
_{i=1,\ldots,N}(\frac{\sum \nolimits_{k=0}^{i-1}E_k}{\sum
\nolimits_{k=1}^{i}L_k} - \alpha)$, which is the optimal solution
according to Table \ref{table1}.

From the above two cases, (\ref{Online1}) coincides with the optimal
solution for $P^*(0)$. Thus, Proposition \ref{proposition:online} is
proved.


\begin{thebibliography}{1}
\bibliographystyle{IEEEbib}


\bibitem{Chen:NetworkES} T. Chen, Y. Yang, H. Zhang, H. Kim, and K. Horneman, ``Network
energy saving technologies for green wireless access networks,''
{\it  IEEE Wireless Commun.}, vol. 18, no. 5, pp. 30-38, Oct. 2011.

\bibitem{Han:GreenRadio} C. Han {\it et al.}, ``Green radio: Radio techniques to enable
energy-efficient wireless networks,'' {\it IEEE Commun. Mag.}, vol.
49, no. 6, pp. 46-54, Jun. 2011.

\bibitem{Niu:CellZooming} Z. Niu, Y. Wu, J. Gong, and Z. Yang, ``Cell zooming for
cost-efficient green cellular networks,'' {\it IEEE Commun. Mag.},
vol. 48, no. 11, pp. 74-78, Nov. 2010.

\bibitem{Bhargava:GreenCellular} Z. Hasan, H. Boostanimehr, and V. Bhargava, ``Green cellular
networks: A survey, some research issues and challenges,'' {\it IEEE
Commun. Surveys \& Tutorials}, vol. 13, no. 4, pp. 524-540, 2011.

\bibitem{YChenComMag} Y. Chen, S. Zhang, S. Xu, and G. Y. Li, ``Fundamental trade-offs
on green wireless networks,'' {\it IEEE Commun. Mag.}, vol. 49, no.
6, pp. 30-37, Jun. 2011.

\bibitem{Energy_EfficientPacketTransmission} E. Uysal-Biyikoglu, B. Prabhakar, and A. Gamal, ``Energy-efficient
packet transmission over a wireless link,'' {\it IEEE/ACM Trans.
Network.}, vol. 10, no. 4, pp. 487-499, Aug. 2002.

\bibitem{Miao} G. W. Miao, N. Himayat, and G. Y. Li, ``Energy-efficient link
adaptation in frequency-selective channels,'' {\it IEEE Trans.
Commun.}, vol. 58, no. 2, pp. 545-554, Feb. 2010.

\bibitem{fractionalprogramming} C. Isheden, Z. Chong, E. Jorswieck, and G. Fettweis, ``Framework for
link-level energy efficiency optimization with informed
transmitter,'' {\it IEEE Trans. Wireless Commun.}, vol. 11, no. 8, pp. 2946-2957, Aug. 2012.

\bibitem{Blume:ES} O. Blume, H. Eckhardt, S. Klein, E. Kuehn, and W. M. Wajda, ``Energy
savings in mobile networks based on adaptation to traffic
statistics,'' {\it Bell Labs Tec. J.}, vol. 15, no. 2, pp. 77-94,
Sep. 2010.

\bibitem{Kim} H. Kim and G. Veciana, ``Leveraging dynamic spare capacity in
wireless system to conserve mobile terminals energy,'' {\it
IEEE/ACM Trans. Network.}, vol. 18, no. 3, pp. 802-815, Jun. 2010.

\bibitem{Kansal} A. Kansal, J. Hsu, S. Zahedi, and M. B. Srivastava, ``Power
management in energy harvesting sensor networks,'' {\it ACM Trans.
Embed. Comput. Syst.}, vol. 6, no. 4, Sep. 2007.

\bibitem{Yang} J. Yang and S. Ulukus, ``Optimal packet scheduling in an energy
harvesting communication system,'' {\it IEEE Trans. Commun.}, vol.
60, no. 1, pp. 220-230, Jan. 2012.

\bibitem{Zhang} C. K. Ho and R. Zhang, ``Optimal energy allocation for wireless
communications with energy harvesting constraints,'' {\it  IEEE
Trans. Sig. Process.}, vol. 60, no. 9, pp. 4808-4818, Sep. 2012.

\bibitem{Yener} K. Tutuncuoglu and A. Yener, ``Optimum transmission policies for
battery limited energy harvesting nodes,'' {\it IEEE Trans. Wireless
Commun.}, vol. 99, no. 2, pp. 1-10, Feb. 2012.


\bibitem{Bai2011}
Q. Bai, J. Li, and J. A. Nossek, ``Throughput maximizing transmission
strategy of energy harvesting nodes,'' in {\it Proc. IWCLD,} Rennes, France, Nov. 2011.

\bibitem{GoldsmithBook}
A. Goldsmith, {\it Wireless Communications.} Cambridge, U.K.: Cambridge Univ. Press, 2005.

\bibitem{Boydbook} S. Boyd and L. Vandenberghe, {\it Convex Optimization}, Cambridge
University Press, 2004.

\bibitem{ZhangCR} R. Zhang, F. Gao, and Y.-C. Liang, ``Cognitive beamforming made
practical: Effective interference channel and learning-throughput
tradeoff,'' {\it IEEE Trans. Commun.}, vol. 58, no. 2, pp. 706-718,
Feb. 2010.

\bibitem{Devillers2012}
B. Devillers and D. Gunduz, ``A general framework for the optimization of
energy harvesting communication systems with battery imperfections,'' {\it Journal of Commun. and Netw.}, vol. 14, no. 2, pp. 130-139, Apr. 2012.


\end{thebibliography}
\end{document}